\documentclass[12pt,a4paper]{article}
\pdfoutput=1
\usepackage{graphics,epsfig,graphicx,color}
\graphicspath{{./Figures/}}
\usepackage{subfigure}
\usepackage{cite}

\usepackage{amssymb,latexsym}

\usepackage{setspace}

\usepackage{amsmath, esint}

\usepackage{mathrsfs,amsfonts}

\usepackage{amsthm,amsxtra}

\usepackage[final]{showkeys}

\usepackage{stmaryrd}

\usepackage{algorithm}
\usepackage{algorithmicx}
\usepackage{algpseudocode}
\usepackage{mathtools}
\newif\ifPDF
\ifx\pdfoutput\undefined
\PDFfalse
\else
\ifnum\pdfoutput > 0
\PDFtrue
\else
\PDFfalse
\fi
\fi

\ifPDF
\usepackage{pdftricks}
\begin{psinputs}
	\usepackage{pstricks}
	\usepackage{pstcol}
	\usepackage{pst-plot}
	\usepackage{pst-tree}
	\usepackage{pst-eps}
	\usepackage{multido}
	\usepackage{pst-node}
	\usepackage{pst-eps}
\end{psinputs}
\else
\usepackage{pstricks}
\fi

\ifPDF
\usepackage[debug,pdftex,colorlinks=true, 
linkcolor=blue, bookmarksopen=false,
plainpages=false,pdfpagelabels]{hyperref}
\else
\usepackage[dvips]{hyperref}
\fi

\usepackage[openbib]{currvita}

\usepackage{fancyhdr}

\newtheorem{theorem}{Theorem}[section]
\newtheorem{lemma}[theorem]{Lemma}
\newtheorem{definition}[theorem]{Definition}

\newtheorem{remark}[theorem]{Remark} 
\newtheorem{corollary}[theorem]{Corollary}

\newcommand{\supp}{\operatorname{supp}}

\newcommand{\eps}{\varepsilon}

\newcommand{\bbR}{\mathbb R} \newcommand{\bbS}{\mathbb S}
\newcommand{\bbZ}{\mathbb Z} 

 \newcommand{\bn}{\mathbf n}

\newcommand{\cA}{\mathcal A} \newcommand{\cB}{\mathcal B}
  
\newcommand{\cE}{\mathcal E} 
 
 \newcommand{\cJ}{\mathcal J}
\newcommand{\cK}{\mathcal K} \newcommand{\cL}{\mathcal L}
 
\newcommand{\cO}{\mathcal O}  
 
\newcommand{\cS}{\mathcal S} 
 
 \newcommand{\cX}{\mathcal X} 
\newcommand{\cY}{\mathcal Y} 

\newcommand{\aver}[1]{\langle {#1} \rangle}

\newenvironment{keywords}
{\noindent{\bf Key words.}\small}{\par\vspace{1ex}}

\makeatletter
\newcommand{\chapterauthor}[1]{%
	{\parindent0pt\vspace*{-25pt}%
		\linespread{1.1}\large\scshape#1%
		\par\nobreak\vspace*{35pt}}
	\@afterheading%
}

\newsavebox\myboxA
\newsavebox\myboxB
\newlength\mylenA

\newcommand*\xoverline[2][0.75]{%
	\sbox{\myboxA}{$\m@th#2$}%
	\setbox\myboxB\null
	\ht\myboxB=\ht\myboxA%
	\dp\myboxB=\dp\myboxA%
	\wd\myboxB=#1\wd\myboxA
	\sbox\myboxB{$\m@th\overline{\copy\myboxB}$}
	\setlength\mylenA{\the\wd\myboxA}
	\addtolength\mylenA{-\the\wd\myboxB}%
	\ifdim\wd\myboxB<\wd\myboxA%
	\rlap{\hskip 0.5\mylenA\usebox\myboxB}{\usebox\myboxA}%
	\else
	\hskip -0.5\mylenA\rlap{\usebox\myboxA}{\hskip 0.5\mylenA\usebox\myboxB}%
	\fi}

\makeatother

\title{Instability of an inverse problem for the stationary radiative transport near the diffusion limit}
\author{Hongkai Zhao, Yimin Zhong}
\date{}
\begin{document}
	\maketitle
	\begin{abstract}
		In this work, we study the instability of an inverse problem of radiative transport equation with angularly independent source and angularly averaged measurement near the diffusion limit, i.e.\ the normalized mean free path (the Knudsen number) $0 < \eps \ll 1$. {For the reconstruction of absorption coefficient,  we show that instability depends on the relative sizes between $\eps$ and the perturbation in measurements.} When $\eps$ is sufficiently small, we obtain exponential instability, which stands for the diffusion regime, and otherwise we obtain H\"{o}lder instability instead, which stands for the transport regime.
	\end{abstract}
\section{Introduction}
In this paper, we study the instability of an inverse problem for stationary radiative transport equation (RTE) near the diffusion limit. The radiative transport equation is the typical model to describe the propagation of radiative particles through a scattering medium. In the stationary setting, we assume the density of particles $u(x, v)$ satisfies following general RTE
\begin{equation}\label{eq:RTE0}
\begin{aligned}
v\cdot \nabla u(x, v) + \sigma_t(x, v)u(x, v) &= \int_{\Omega} k(x, v,v') u(x, v') d\mu(v'),\quad &\text{ in }& D\times \Omega \\
u(x, v) &=f(x, v), \quad &\text{ on }&\Gamma_{-}
\end{aligned}
\end{equation}
The spatial domain $D\subset \bbR^d, d\ge 2$ is bounded with convex smooth boundary, $\Omega = \bbS^{d-1}$ denotes the unit sphere surface in $\bbR^d$, $d\mu$ is the associated uniform probability measure of $\Omega$. $f(x, v)$ models an incident density of particles entering the domain. The incoming and outgoing boundary sets $\Gamma_{-}$ and $\Gamma_{+}$ are defined by $\Gamma_{\pm} = \{(x, v)\in \partial D\times \Omega: \pm n_x \cdot  v > 0\}$ respectively,  where $n_x$ is the outward unit normal vector at $x\in\partial D$.  The optical parameters $\sigma_{t}(x, v)$ and $k(x, v,v')$ are the total absorption and scattering coefficients respectively. For most inverse transport problems, these optical parameters are unknown and needed to be reconstructed from certain boundary or interior measurements ~\cite{gao2010multilevel,bal2009inverse,ren2015inverse,stefanov2008inverse,bal2016ultrasound,choulli1999inverse}. Such inverse problems have a wide range of applications in medical imaging, remote sensing, nuclear engineering, astrophysics, etc., we refer the interested readers to e.g.~\cite{arridge2009optical,ren2007transport,leblond2009early,tsang1985theory,tucker1986satellite,mccormick1992inverse,glasstone2012nuclear}.

In practice, it is common to assume the optical parameters to be independent of the angular variable, which means $\sigma_{t}(x, v) = \sigma_t(x)$ and $k(x, v, v') = \sigma_s(x)p(v,v')$ for some  \emph{a-priori} known phase function $p$. {The most used measurement is the so-called \emph{albedo} operator defined by
\begin{equation}
\cA : u|_{\Gamma_{-}} \mapsto u|_{\Gamma_{+}},
\end{equation}
where the source function $u|_{\Gamma_{-}}$ and the measurement of $u|_{\Gamma_+}$  are both assumed to be \emph{angularly resolved}. The relevant theories for $\cA$ have been extensively studied in~\cite{bal2008stability, bal2010stability,lai2018inverse} using the singular decomposition of Schwartz kernels, the coefficients $\sigma_t(x)$ and $\sigma_s(x)$ can be \emph{both} reconstructed with H\"{o}lder type stability for $d\ge 3$. In~\cite{lai2018inverse}, the authors have shown the stability estimate actually transits from the H\"{o}lder type to logarithmic type when the Knudsen number is approaching the diffusion limit. }

In many applications, however, a full knowledge of the albedo operator, which requires sufficient sampling and accurate measuring of angular dependent data,  is either too expensive or impossible. Furthermore, the angular resolved data may suffer from very low particle counts in certain directions, which leads to low signal to noise ratio in the measurements. Therefore, the incident sources and the measurements are \emph{both} angularly independent often in practice. {A typical angularly averaged albedo operator is the following
\begin{equation}
\cB: u|_{\Gamma_{-}}(x) \mapsto u|_{\Gamma_{+}}(x)=\int_{n_x \cdot v > 0} v\cdot n_x u|_{\Gamma_{+}}(x, v) d\mu(v)
\end{equation}
with $u|_{\Gamma_{-}}\in L^p(\partial D)$, $p\ge 1$. Without the angular dependences, the singular decomposition technique is not applicable anymore.
It has been widely accepted that such inverse problem is quite ill-posed. Some synthetic numerical examples~\cite{bi2015image,tang2017mixed} are performed to verify the ill-posedness. However, the best known uniqueness and stability estimate results are only limited to the linearized case with small optical parameters~\cite{bal2009inverse,bal2008inverse2}, where the multiple-scattering component is dominated by the single-scattering component. The complete theory for uniqueness and stability estimate are still unavailable. }

One interesting setting is to use the \emph{angularly resolved} sources and the \emph{angularly averaged} measurement on $\Gamma_{+}$, this can be viewed as some kind of mixed mapping between $\cA$ and $\cB$. Such problem has been studied in~\cite{langmore2008stationary,chen2018stability}. Since the incident sources still contain the angular dependences, one can observe that the singular decomposition of Schwartz kernel still works for the ballistic part, which allows using X-ray transform to recover $\sigma_t$ uniquely. {Nevertheless, the reconstruction of scattering coefficient is more difficult due to the failure to distinguish between single-scattering and multiple-scattering. Under the linearized setting, the authors in~\cite{chen2018stability} have analyzed the instability of the reconstructions for both $\sigma_t$ and $\sigma_s$ when the Knudsen number $\eps$ is approaching the diffusion limit.}

In this paper, we consider the inverse transport problem for the angularly averaged albedo operator $\cB$ and study the instability of reconstructing the absorption coefficient $\sigma_a(x) := \sigma_t(x) - \sigma_s(x)$ near the diffusion regime. To characterize the closeness to diffusion approximation, we rescale the RTE by the Knudsen number $0 < \eps \ll 1$ as follows 
\begin{equation}\label{eq:RTE}
\begin{aligned}
v\cdot \nabla u(x, v) + \left(\eps \sigma_a (x) + \frac{1}{\eps} \sigma_s(x) \right)u(x, v) &= \frac{1}{\eps}\sigma_s(x) \aver{u},\quad &\text{ in }& D\times \Omega \\
u(x, v) &=f(x), \quad &\text{ on }&\Gamma_{-}
\end{aligned}
\end{equation}
In this study, we assume the scattering coefficient $\sigma_s(x)p(v,v')$ is known for $x\in D$. For simplicity, we let the phase function $p(v,v')\equiv 1$ and use the symbol $\aver{u}$ to represent the following angularly averaged integral, which is independent of $v$,
\begin{equation}
\aver{u} = \int_{\Omega} p(v,v')u(x,v')d\mu(v') =\int_{\Omega} u(x, v) d\mu(v).
\end{equation}
We also define the \emph{scaled} measurement by the averaged albedo operator in the following
\begin{equation}
\begin{aligned}
&\Lambda_{\sigma_a} : f(x) \in \cX \mapsto \Lambda_{\sigma_a} f(x) = \cJ_{+}(x) \in \cY\\
&\cJ_{+}(x)= \frac{1}{\eps}\int_{n_x \cdot v > 0} v\cdot n_x u(x, v) d\mu(v),
\end{aligned}
\end{equation}
where $u(x, v)$ solves the RTE~\eqref{eq:RTE} and the Banach spaces $\cX$, $\cY$ will be clarified later. Since the incoming boundary condition $u(x, v) = f(x)$ on $\Gamma_{-}$ is already provided, therefore we replace the outward boundary current $\cJ_{+}(x)$ by the following total boundary current,
\begin{equation}
\cJ (x)= \frac{1}{\eps}\int_{\Omega} v\cdot n_x u(x, v) d\mu(v), \quad x\in \partial D.
\end{equation}
and redefine $\Lambda_{\sigma_a}$ by
\begin{equation}
\begin{aligned}
&\Lambda_{\sigma_a} : f(x) \in \cX \mapsto \Lambda_{\sigma_a} f(x) = \cJ(x) \in \cY
\end{aligned}
\end{equation}

For the above angularly averaged albedo operator $\Lambda_{\sigma_a}$, there are two limiting values of $\eps$. When $\eps\to\infty$, we arrive at the purely linear transport equation, where we can drop the collision term by $\sigma_s\equiv 0$. Then the measurement will be exactly the line Radon transforms of the absorption coefficient $\sigma_a$. In this setting, the inverse problem has the H\"{o}lder type stability since line Radon transform only loses one half of derivative. Another special case is the diffusion limit of~\eqref{eq:RTE} with $\eps\rightarrow 0$, using the Hilbert asymptotic expansion, the RTE is then approximated by the following diffusion equation~\cite{dautray1993mathematical, papanicolaou1975asymptotic,bensoussan1979boundary,papanicolaou1995diffusion}
\begin{equation}
-\frac{1}{d}\nabla \left( \frac{1}{\sigma_s(x)} U(x)\right) + \sigma_a(x) U(x) = 0
\end{equation}
and the angularly averaged albedo operator $\Lambda_{\sigma_a}$ turns to be the Dirichlet-to-Neumann (DtN) map.
The reconstruction of the potential of Schr\"{o}dinger equation from DtN map is closely related to the electrical impedance tomography (EIT) or Calder\'{o}n's problem. {The EIT problem was studied extensively in recent decades \cite{uhlmann2009electrical,sylvester1987global,nachman1996global, alessandrini1991determining}. It is well-known that the reconstruction of isotropic conductivity from the DtN map is severely ill-posed, the reconstruction has \emph{both} sharp logarithmic stability and exponential instability~\cite{uhlmann2009electrical,sylvester1987global,mandache2001exponential,di2003examples}.}

Regarding the above two cases, the stability estimate of $\sigma_a$ transits from the H\"{o}lder type ($\eps \to \infty$) to the logarithmic type ($\eps\to 0$). { Such transition was studied recently for the angularly resolved albedo operator $\cA$ in~\cite{lai2018inverse} based on singular decomposition of the Schwartz kernel. While for the angularly averaged albedo operator $\Lambda_{\sigma_a}$, the transition is still not well understood.} In general, there are no uniqueness or stability estimate results for this problem. Under the linearized setting, the author in~\cite{bal2008inverse} have proved a stability estimate of the scattering coefficient with smallness assumption.

In the following context, we fix $\cX = H^s(\partial D)$ and $\cY = H^{-s}(\partial D)$ with parameter $s > \frac{d+4}{2}$ for the angularly averaged albedo operator $\Lambda_{\sigma_a}$. We use a constructive approach to show that the instability estimate of $\Lambda_{\sigma_a}: H^s(\partial D)\to H^{-s}(\partial D)$ varies from the H\"{o}lder type to the exponential type as the Knudsen number $\eps$ decreases to $0$. {The framework of our study is motivated by~\cite{mandache2001exponential} which studied the instability for potential reconstruction for Schr\"{o}dinger equation. }

The rest of this paper is organized as follows. In the Section~\ref{sec:MAIN}, we make appropriate assumptions on the coefficients and state our main results.  In Section~\ref{sec:PRELIM}, we introduce some preliminary results and provide key proofs. In Section~\ref{sec:ESTIMATE}, we provide an estimate for the matrix representation for $\Lambda_{\sigma_a}$. In the Section~\ref{sec:INSTAB}, we prove the main results by using Kolmogorov's entropy theory. We give conclusions in Section~\ref{sec:CONCL} and provide the proofs of two Lemmas used in the proof of our main result in Appendix.

\section{Main results}\label{sec:MAIN}
For our constructive approach, we fix the domain $D = B(0, 1)\subset \bbR^{d}$, where $B(z, r)$ denotes a ball centered at $z$ with radius $r$. Assume the scattering coefficient $\sigma_s(x)\equiv \sigma_s$ is a positive constant over $D$. We define the admissible set of the absorption coefficient by
 $$\cS := \{\sigma_a \,|\, \sigma_a(x) \in L^{\infty}(D),\, \sigma_a \ge 0,\, \supp \sigma_a \subset K, \sigma_a \in C^{q}(K)\}$$ where the interior region $K = B(0, r_0)$ with $0 < r_0 < 1$ and $q > 0$. We also introduce the $d$-dimensional \emph{complex} spherical harmonic basis $\mathbb{H}^d:=\{ Y_{mj} \,|\, m\ge 0, 1\le j \le p_m \}$ on the unit sphere $\bbS^{d-1}$,  where $Y_{mj}$ being a spherical harmonic of order $m$ and $p_m = \binom{m+d-1}{d-1} - \binom{m+d-3}{d-1} \le 2 (1+m)^{d-2}$. $\mathbb{H}^d$ forms a Schauder basis for both $\cX=H^{s}(\partial D)$ and $\cY= H^{-s}(\partial D)$.

\begin{theorem}\label{thm:MAIN}
	For any $q > 0$ and dimension $d\ge 2$, any  $s > \frac{d+4}{2}$ and $R > 0$, there is a constant $\beta > 0$ such that for any $\theta\in (0, \frac{R}{2})$ and $\sigma_{a, 0}\in L^{\infty}(D)$ with $\|\sigma_{a, 0}\|_{\infty}\le \frac{R}{2}$, $\supp \sigma_{a, 0} \subset K = B(0, r_0)$, there are absorption coefficients $\sigma_{a, 1}, \sigma_{a, 2}\in C^q(K)$, also supported in $K$ such that 
	\begin{equation}
	\begin{aligned}
	&\|\Lambda_{1} - \Lambda_2\|_{H^s(\partial D)\to H^{-s}(\partial D)} \le 8\sqrt{2} \omega(\theta^{-\frac{d}{(2d+1)q}}),\\
	&\|\sigma_{a, 1} - \sigma_{a, 2}\|_{\infty} \ge \theta, \\
	&\|\sigma_{a, i} - \sigma_{a, 0}\|_{C^q} \le \beta,\quad i = 1,2 \\
	&\|\sigma_{a, i} - \sigma_{a, 0}\|_{\infty} \le \theta,\quad i=1,2
	\end{aligned}
	\end{equation}
	where $\Lambda_1, \Lambda_2$ are the averaged albedo operators for $\sigma_{a,1},\sigma_{a,2}$ respectively,  $\omega(t)$ is the solution to the following equation 
	\begin{equation}
	t = \log(\omega^{-1}) + \frac{\eps}{\omega} + \left(\frac{\eps}{\omega}\right)^{-1/\tau},\quad \tau = \frac{d+4}{2} -s < 0.
	\end{equation} 
\end{theorem}
Let $s = \frac{d+4}{2} + 1$, then $\tau = -1$, depending on the relation of $\eps$ and $\theta$, we introduce two cases as follows.
\begin{corollary}\label{cor:EXP}
	When Knudsen number $\eps$ is small enough such that $$\eps \le \frac{1}{3}\theta^{-\frac{d}{(2d+1)q}}\exp\left(-\frac{1}{3}\theta^{-\frac{d}{(2d+1)q}}\right),$$ then under the same assumptions of Theorem~\ref{thm:MAIN},
	\begin{equation}
	\begin{aligned}
	&\|\Lambda_{1} - \Lambda_2\|_{H^s(\partial D)\to H^{-s}(\partial D)} \le 8\sqrt{2} \exp\left(-\frac{1}{3}\theta^{-\frac{d}{(2d+1)q}}\right),\\
	&\|\sigma_{a, 1} - \sigma_{a, 2}\|_{\infty} \ge \theta, \\
	&\|\sigma_{a, i} - \sigma_{a, 0}\|_{C^q} \le \beta,\quad i = 1,2 \\
	&\|\sigma_{a, i} - \sigma_{a, 0}\|_{\infty} \le \theta,\quad i=1,2
	\end{aligned}
	\end{equation}
\end{corollary}

\begin{corollary}\label{cor:HOLDER}
	When Knudsen number $\eps$ satisfies
	$$1\gg\eps > \frac{1}{3}\theta^{-\frac{d}{(2d+1)q}}\exp\left(-\frac{1}{3}\theta^{-\frac{d}{(2d+1)q}}\right),$$
	then under the same assumptions of Theorem~\ref{thm:MAIN},
	\begin{equation}
	\begin{aligned}
	&\|\Lambda_{1} - \Lambda_2\|_{H^s(\partial D)\to H^{-s}(\partial D)} \le 24 \sqrt{2}\eps\theta^{\frac{d}{(2d+1)q}},\\
	&\|\sigma_{a, 1} - \sigma_{a, 2}\|_{\infty} \ge \theta, \\
	&\|\sigma_{a, i} - \sigma_{a, 0}\|_{C^q} \le \beta,\quad i = 1,2 \\
	&\|\sigma_{a, i} - \sigma_{a, 0}\|_{\infty} \le \theta,\quad i=1,2
	\end{aligned}
	\end{equation}
\end{corollary}
\begin{remark}
The above two Corollaries indicate the transition from the H\"{o}lder type instability in transport regime to the exponential type instability in diffusion limit as $\eps$ becomes small enough.
\end{remark}
\section{Preliminaries}\label{sec:PRELIM}
In this section, we present the following preliminary results for the scaled RTE~\eqref{eq:RTE}. 
\begin{lemma}\label{lem:CURRENT}
	If the boundary source $f \in L^p(\partial D)$ for  $p\ge 1$, then $u(x, v)\in L^p(D\times\Omega)$ and $\aver{u}\in L^p(D)$.
\end{lemma}
\begin{proof}
	This lemma is a direct conclusion from Theorem 3.3 in~\cite{agoshkov2012boundary}.
\end{proof}
\begin{lemma}\label{lem:DTN}
	If the boundary source $f\in L^2(\partial D)$ {and $s > \frac{d+4}{2}$}, then $\Lambda_{\sigma_a} f\in H^{-1/2}(\partial D)\subset H^{-s}(\partial D)$.
\end{lemma}
\begin{proof}
	To see $\psi = \Lambda_{\sigma_a} f\in H^{-1/2}(\partial D)$, we take an arbitrary function $h\in H^{1/2}(\partial D)$, then by Fubini's theorem, 
	\begin{equation}
	\left| \int_{\partial D} \psi(x) h(x) dS(x)\right|  =  \left|  \int_{\Omega} \int_{\partial D} v\cdot n_x u(x,v)h(x)  dS(x) d\mu(v)\right|.
	\end{equation}
	On the other hand, $h\in H^{1/2}(\partial D)$, then there exists a linear bounded extension operator $\cE: H^{1/2}(\partial D)\to H^{1}(D)$ with trace of $\cE h = h$ on $\partial D$. We use $\hat{h}$ to denote the extension. Using integration by parts, 
	\begin{equation}\nonumber
	\begin{aligned}
	\int_{\partial D} v\cdot n_x u(x,v)h(x)  dS(x)  =&\, \int_{D} v\cdot \nabla \hat{h}(x) u(x, v) dx \\& -\int_{D}\left(\eps \sigma_a + \frac{1}{\eps} \sigma_s(x) \right)u(x, v)  \hat{h}(x) dx +\frac{1}{\eps} \int_{D}\sigma_s(x) \aver{u}  \hat{h}(x) dx,
	\end{aligned}
	\end{equation}
	then we have following estimate
	\begin{equation}
	\begin{aligned}
	\left| \int_{\partial D} \psi(x) h(x) dS(x)\right| \le \left| \int_{\Omega} \int_D v\cdot \nabla  \hat{h}(x) u(x, v) dx d\mu(v) \right| + \left|   \int_D \eps \sigma_a(x)  \hat{h}(x)\aver{u} dx  \right|.
	\end{aligned}
	\end{equation}
	By the Cauchy-Schwartz inequality, there exists a constant $C, \tilde{C}>0$ such that
	\begin{equation}\nonumber
	\left| \int_{\partial D} \psi(x) h(x) dS(x)\right|\le C\left( \|u\|_{L^2(D\times\Omega)} \|\nabla \hat{h}\|_{L^2(D)} + \|\aver{u}\|_{L^2(D)}\|\hat{h}\|_{L^2(D)}\right)\le \tilde{C}\|f\|_{L^2(\partial D)}\|h\|_{H^{1/2}(\partial D)}.
	\end{equation}
\end{proof}

\begin{lemma}\label{lem:ADJ}
	If  $u$ is the solution to the RTE~\eqref{eq:RTE} and $w$ satisfies the following adjoint radiative transfer equation with outgoing boundary condition,
	\begin{equation}\label{eq:ADJ}
	\begin{aligned}
	-v\cdot \nabla w + (\eps \sigma_a + \frac{1}{\eps}\sigma_s) w &= \frac{1}{\eps}\sigma_s \aver{w}\quad&\text{ in }& D\times\Omega,\\ 
	w(x, v) &= g(x),\quad&\text{ on }&\Gamma_{+}.
	\end{aligned}
	\end{equation}
	then $u$ and $w$ satisfy the following relation,
	\begin{equation}
	\int_{\partial D} \left(\int_{\Omega} n_x \cdot v u(x, v)  d\mu(v) \right)g(x) dS(x) = -\int_{\partial D} \left(\int_{\Omega} n_x\cdot v w(x, v) d\mu(v)\right)  f(x) dS(x).
	\end{equation}
\end{lemma}
\begin{proof}
	The equality is obvious by divergence theorem.
\end{proof}

\section{The basic estimate}\label{sec:ESTIMATE}
Let $u_0(x, v)$ be the solution to the following radiative transport equation with zero absorption coefficient and denote the associated measurement operator by $\Lambda_0$. 
\begin{equation}\label{eq:BASE}
\begin{aligned}
v\cdot \nabla u_0 + \frac{1}{\eps} \sigma_s(x) u_0(x, v) &= \frac{1}{\eps} \sigma_s(x) \aver{u_0},\quad &\text{ in }& D\times \Omega,\\
u_0(x,v) &= f(x),\quad&\text{ on }& \Gamma_{-}.
\end{aligned}
\end{equation}
Assume $u(x, v)$ be the solution to RTE~\eqref{eq:RTE} and $\phi = u - u_0$, then $\phi$ satisfies the following RTE with vacuum incoming boundary condition,
\begin{equation}\label{eq:PERTURB}
\begin{aligned}
v\cdot \nabla \phi(x,v) + \left(\eps \sigma_a(x) + \frac{1}{\eps}\sigma_s(x) \right)\phi(x, v) &= \frac{1}{\eps} \sigma_s(x)\aver{\phi} - \eps \sigma_a(x) u_0,\quad &\text{ in }& D\times\Omega, \\
\phi(x, v) &= 0,\quad &\text{ on }&\Gamma_{-}.
\end{aligned}
\end{equation} 
For each admissible $\sigma_a(x)$, we define the linear operator $\Gamma(\sigma_{a}) := \Lambda_{\sigma_{a}} - \Lambda_0\in \cL(\cX, \cY)$. For any $f, g\in H^{s}(\partial D)$, we have the following equality.
\begin{equation}\label{eq:ENTRIES}
\begin{aligned}
\langle \Gamma(\sigma_a) f, g\rangle &= \frac{1}{\eps} \int_{\partial D \times \Omega} v\cdot n_x \phi(x, v) \overline{g(x)} d\mu(v) dS(x) \\
&=  - \int_K \sigma_a(x) \overline{\hat{g}(x)} \aver{u}(x)dx + \frac{1}{\eps} \int_{D} \nabla \overline{\hat{g}(x)} \cdot \left( \int_{\Omega} v \phi(x, v) d\mu(v)\right) dx
\end{aligned}
\end{equation}
where $\hat{g}(x) \in H^1(D)$ is an arbitrary extension of $g(x)$ in $D$ and $\overline{g(x)}$ and $\overline{\hat{g}(x)}$ are the complex conjugates of $g(x)$ and $\hat{g}(x)$ respectively. Similarly, the above quantity can also be computed by the adjoint RTE using Lemma~\ref{lem:ADJ}.
\begin{equation}\label{eq:ENTRIES2}
\langle \Gamma(\sigma_a) f, g\rangle = \int_{K} \sigma_a(x) \hat{f}(x)\overline{\aver{w}(x)}dx - \frac{1}{\eps} \int_D \nabla \hat{f}(x) \cdot \int_{\Omega} v\overline{\varphi(x,v)} d\mu(v) dx,
\end{equation}
where $w_0$ and $w$ are the solutions to the following adjoint radiative transfer equations,
\begin{equation}\label{eq:PERTURB2}
\begin{aligned}
-v\cdot \nabla w_{0}+ \frac{1}{\eps} \sigma_s(x) (x) w_{0}(x, v) &= \frac{1}{\eps} \sigma_s(x) \aver{w_{0}},\quad &\text{ in }& D\times \Omega,\\
-v\cdot \nabla w_{} + (\eps \sigma_a(x) + \frac{1}{\eps} \sigma_s (x) )w_{}(x, v) &= \frac{1}{\eps} \sigma_s(x) \aver{w_{}},\quad &\text{ in }& D\times \Omega,\\
w_{0}(x,v) = w_{}(x, v) &= g(x),\quad&\text{ on }& \Gamma_{+}.
\end{aligned}
\end{equation}
and $\varphi = w_{} - w_{0}$, $\hat{f}\in H^1(D)$ is an arbitrary extension of $f(x)$ in $D$, $\overline{\aver{w}(x)}$ and $\overline{\varphi(x, v)}$ are the complex conjugates of $\aver{w}(x)$ and $\aver{\varphi(x, v)}$ respectively. 
 Combining the above two representations for $\langle \Gamma(\sigma_a) f, g\rangle$, we introduce a basic estimate in the following lemma.
\begin{lemma}\label{lem:ESTIMATE}
	There is a constant $C_0 = C_0(r_0, d, s)$ such that for any $4$-tuple $(m,j,n,k)$ that $m, n\ge 0$ and $j \le p_m$, $k\le p_n$, 
	\begin{equation}
|\langle \Gamma(\sigma_a) Y_{mj}, Y_{nk}\rangle | \le C_0 \|\sigma_a\|_{\infty} (1 +l)  \left( r_0^{l} + \eps (1 + l)  \right).
	\end{equation}
	where $l = \max(m,n)$.
\end{lemma}
\begin{proof}
	The left-hand-side has the following representations by~\eqref{eq:ENTRIES} and~\eqref{eq:ENTRIES2}. Each representation has two parts, we denote by $I_{i, mjnk}$ and $L_{i,mjnk}$, $i=1,2$, respectively.
	\begin{equation}\nonumber
	\begin{aligned}
		\langle \Gamma(\sigma_a) Y_{mj}, Y_{nk}\rangle  &=- \int_K \sigma_a(x) \overline{\hat{Y}_{nk}(x)} \aver{u_{mj}}(x)dx + \frac{1}{\eps} \int_{D} \nabla \overline{\hat{Y}_{nk}(x)} \cdot \int_{\Omega} v \phi_{mj}(x, v) d\mu(v) dx \\
		&= I_{1, mjnk} + I_{2,mjnk}.\\
		\langle \Gamma(\sigma_a) Y_{mj}, Y_{nk}\rangle  &= \int_{\Omega} \sigma_a(x) \hat{Y}_{mj}(x)\overline{\aver{w_{ nk}}(x)}dx - \frac{1}{\eps} \int_D \nabla \hat{Y}_{mj}(x) \cdot \int_{\Omega} v\overline{\varphi_{nk}(x,v)} d\mu(v) dx \\
		 &= L_{1, mjnk} + L_{2,  mjnk},
	\end{aligned}
	\end{equation}
	The functions $\hat{Y}_{mj}$ and $\hat{Y}_{nk}$ are arbitrary $H^1(D)$ extensions of spherical harmonics $Y_{mj}$ and $Y_{nk}$ respectively. For the first representation, $u_{mj}$ is the solution to~\eqref{eq:RTE} with $f = Y_{mj}$. Denote $u_{0,mj}$ be the solution to~\eqref{eq:BASE} with $f = Y_{mj}$, the function $\phi_{mj}$ satisfies the following RTE,
	\begin{equation}
	\begin{aligned}
	v\cdot \nabla \phi_{mj} + \left(\eps \sigma_a(x) + \frac{1}{\eps}\sigma_s(x)\right) \phi_{mj}&= \frac{1}{\eps} \sigma_s(x)\aver{\phi_{mj}} - \eps \sigma_a(x) u_{0,mj},\quad &\text{ in }& D\times\Omega, \\
	\phi_{mj}(x, v) &= 0,\quad &\text{ on }&\Gamma_{-}.
	\end{aligned}
	\end{equation} 
	For the second representation, denote $w_{0, nk}$  and $w_{nk}$ be the solutions to~\eqref{eq:PERTURB2} with $g = Y_{nk}$, the function $\varphi_{nk}$ satisfies the following adjoint RTE,
	\begin{equation}
	\begin{aligned}
		-v\cdot\nabla \varphi_{nk} + \left(\eps \sigma_a(x) + \frac{1}{\eps}\sigma_s(x)  \right) \varphi_{nk} &= \frac{1}{\eps} \sigma_s(x)\aver{\varphi_{nk}} - \eps \sigma_a(x) w_{0,nk},\quad &\text{ in }& D\times\Omega, \\
		\varphi_{nk}(x, v) &= 0,\quad &\text{ on }&\Gamma_{+}.
	\end{aligned}
	\end{equation}
For each $4$-tuple $(m,j, n, k)$, it is obvious that the quantity $|\langle \Gamma(\sigma_a) Y_{mj}, Y_{nk}\rangle|$ is bounded by
\begin{equation}
|\langle \Gamma(\sigma_a) Y_{mj}, Y_{nk}\rangle| \le \min\left( |I_{1,mjnk}| +|I_{2,mjnk}|, |L_{1,mjnk}| + |L_{2,mjnk}| \right).
\end{equation}
If we define $M_{i,mjnk}$ ($i=1,2$) as 
\begin{equation}
M_{i,mjnk} = \begin{cases}
L_{i,mjnk},\quad &m\ge n ,\\
I_{i,mjnk}, \quad &n > m,
\end{cases}
\end{equation}
$|\langle \Gamma(\sigma_a) Y_{mj}, Y_{nk}\rangle|\le |M_{1,mjnk}| + |M_{2,mjnk}|$. Hence it suffices to prove two inequalities as follows.
\begin{enumerate}
	\item $|M_{1,mjnk}|\le C_0 \|\sigma_a\|_{\infty} (1 + \max(m, n))r_0^{\max(m, n)}$.
	\item $|M_{2,mjnk}|\le C_0 \|\sigma_a\|_{\infty}\eps (1 + \max(m, n))^2$.
\end{enumerate}
\noindent {\bf{Estimate of $M_{1,mjnk}$.}} When $n > m$, $M_{1,mjnk} = I_{1,mjnk}$. Since the spherical harmonic $Y_{nk}$ can be naturally extended to a harmonic function by $\hat{Y}_{nk}(x) = |x|^n Y_{nk}(x/|x|)$, we obtain the following basic estimate by Cauchy-Schwartz inequality,
\begin{equation}
\begin{aligned}
|I_{1,mjnk}| &= \left| \int_{K} \sigma_a(x) |x|^n \overline{Y_{nk}(x/|x|) }\aver{u_{mj}}(x) dx \right| \\&\le \|\sigma_a\|_{\infty} \left( \int_0^{r_0} r^{2n + d - 1} dr \right)^{1/2} \|\aver{u_{mj}}\|_{K}\\
&\le \|\sigma_a\|_{\infty}\frac{1}{\sqrt{2n + d}} r_0^{n + d/2} \|\aver{u_{mj}}\|_K. 
\end{aligned}
\end{equation}
To estimate $\|\aver{u_{mj}}\|_{K}$,  we consider the Peierls integral equation for $\aver{u_{mj}}$~\cite{ren2016fast,vladimirov1963mathematical,agoshkov2012boundary}, 
\begin{equation}
\aver{u_{mj}}(x) = \int_{D} \cK(x, y) \frac{\sigma_s}{\eps} \aver{u_{mj}}(y) dy + \int_{\partial D} \cK(x, y) \frac{y - x}{|y-x|}\cdot n_y {Y}_{mj}(y) dS(y),
\end{equation}
where the integral kernel $\cK$ is 
\begin{equation}
\begin{aligned}
\cK(x, y) &= \frac{1}{\nu_d}\frac{E(x, y)}{|x-y|^{d-1}},\\
E(x, y) &= \exp\left( -\frac{|x-y|}{\eps}\int_{0}^1 (\eps^2\sigma_a  + \sigma_s)(x + t(y-x)) dt\right),
\end{aligned}
\end{equation}
and $\nu_d$ is the area of the unit sphere $\bbS^{d-1}$. If we denote integral operators $\cK_1$ and $\cK_2$ as
\begin{equation}\label{eq:Ks}
\begin{aligned}
\cK_1 f(x) &= \int_D \cK(x, y)  \frac{\sigma_s}{\eps} f(y)dy,\\
\cK_2 f(x) &= \int_{\partial D} \cK(x,y)\frac{y-x}{|y-x|}\cdot n_y f(y) dS(y)
\end{aligned}
\end{equation}
then the equation is solved by
\begin{equation}
\aver{u_{mj}} = (I - \cK_1)^{-1} \cK_2 Y_{mj}.
\end{equation}
Using Lemma~\ref{lem:A1}, we have
\begin{equation}
\| \aver{u_{mj}} \|_K \le \|\aver{u_{mj}}\|_{D}  \le C \|Y_{mj}\|_{H^{3/2}(\partial D)}\le C (1 + m)^{3/2}.
\end{equation}
The constant $C$ only depends on $\sigma_s$ and $d$, which implies $$|I_{1,mjnk}|\le \|\sigma_a\|_{\infty} \frac{C(1 + m)^{3/2}}{\sqrt{2n + d}} r_0^{n+d/2}\le C \|\sigma_a\|_{\infty}(1+n) r_0^n.$$ On the other hand, when $m \ge n$, $M_{1,mjnk} = L_{1,mjnk}$, using the equality in Lemma~\eqref{lem:ADJ} for adjoint RTE and following the above argument, we obtain the similar estimate for $|L_{1,mjnk}|$,
\begin{equation}
|L_{1,mjnk}| \le C \|\sigma_{a}\|_{\infty}(1+m) r_0^{m},
\end{equation}
and therefore $|M_{1,mjnk}|\le C \|\sigma_a\|_{\infty}  (1+\max(m,n))r_0^{\max(m,n)}$.
\\\\
\noindent{\bf{Estimate of $M_{2,mjnk}$.}} When $n > m$,  $M_{2, mjnk} = I_{2, mjnk}$. Since $\eps \ll 1$, we can estimate the solution $u_{0, m,j}$ to RTE~\eqref{eq:BASE} using the asymptotic expansion introduced in~\cite{Wu2015}. The asymptotic analysis of the solution to~\eqref{eq:BASE} has been studied in~\cite{Li2017,Wu2015,bensoussan1979boundary} and the references therein. For a general boundary source depending on the angular variable $f=f(x, v)$ in~\eqref{eq:BASE}, the asymptotic expansion in~\cite{bensoussan1979boundary} could potentially capture the incorrect behavior in the boundary layer~\cite{Wu2015}, and the interior solution's behavior of the asymptotic expansion of~\cite{bensoussan1979boundary} is characterized in~\cite{Li2017}. In our case, $f = Y_{mj}(x)$ is independent of angular variable and $C^{\infty}$ smooth.
Let $\tilde{f}$ be the harmonic extension of $f$, by Lemma~\ref{lem:A1}, we have 
\begin{equation}
    \aver{u_{0, mj}} = \tilde{f} - (I -\cK_1)^{-1} \int_{\partial D} G(|x-y|) \frac{\partial\tilde{f}}{\partial n} dS(y),
\end{equation}
where $G$ is the integral kernel $$	G(r) = -\frac{1}{\nu_d}\int_{r}^{\infty} \frac{e^{-\frac{\sigma_s}{\eps}\rho}}{\rho^{d-1}} d\rho. $$ 
Since we have $\|(I - K_1)^{-1}\|_{L^2(D)\to L^2(D)} = \cO(\eps^2)$ and $G: L^2(\partial D)\to L^2(D)$ is bounded by the single layer potential operator, 
\begin{equation}
   \left\| \int_{\partial D} G(|x-y|) \frac{\partial\tilde{f}}{\partial n} dS(y) \right\|_{L^2(D)} \le C \|\frac{\partial\tilde{f}}{\partial n}\|_{L^2(\partial D)}.
\end{equation}
Therefore $     \aver{u_{0, mj}} = \tilde{f} + h$ that $\|h\|_{L^2(D)}= \cO(\eps^2) \|\frac{\partial\tilde{f}}{\partial n}\|_{L^2(\partial D)}$. Bring into the equation~\eqref{eq:BASE}, we solve 
\begin{equation}
    u_{0,mj}(x, v) = u_{0,mj}^B(x,v) + \int_0^{\tau_{-}(x,v)} e^{-\frac{\sigma_s}{\eps} l} \frac{\sigma_s}{\eps}(\tilde{f}(x-lv) + h(x-lv)) dl 
\end{equation}
where $u_{0,mj}^B = f(x-\tau_{-}(x,v) v) e^{-\frac{\sigma_s}{\eps}\tau_{-}(x,v)}$ which is exponentially small (in $L^2(K\times\Omega)$ sense) for $x\in K$. The interior part 
\begin{equation}
\begin{aligned}
    u_{0,mj}^I(x,v) &= \int_0^{\tau_{-}(x,v)} e^{-\frac{\sigma_s}{\eps} l} \frac{\sigma_s}{\eps}(\tilde{f}(x-lv) + h(x-lv)) dl \\
    &=\int_0^{\tau_{-}(x,v)} e^{-\frac{\sigma_s}{\eps} l} \frac{\sigma_s}{\eps}\left( \tilde{f}(x) - lv\cdot \nabla \tilde{f} + \frac{l^2}{2}  v^{\otimes 2} \nabla^2\tilde{f}  + \cdots \right) dl \\&\quad + \int_0^{\tau_{-}(x,v)} e^{-\frac{\sigma_s}{\eps} l} \frac{\sigma_s}{\eps} h(x-lv) dl,
\end{aligned}
\end{equation}
where the integral operator $T: L^2(D)\to L^2(D\times\Omega)$ that $T h := \int_0^{\tau_{-}(x,v)} e^{-\frac{\sigma_s}{\eps} l} \frac{\sigma_s}{\eps} h(x-lv) dl$ is uniformly bounded due to Lemma~\ref{lem:A2}. Therefore the asymptotic expansion for $u_{0,mj}^I$ has the form (in $L^2(K\times\Omega)$ sense)
\begin{equation}\label{EQ: MJ}
\|u_{0,mj}^I (x, v) - \tilde{f}(x) + \frac{\eps}{\sigma_s} v\cdot \nabla \tilde{f}(x)\| = \cO(\eps^2)\left(\|\nabla^2 \tilde{f}\|_{L^{2}(D)} + \|\frac{\partial\tilde{f}}{\partial n}\|_{L^2(\partial D)}\right),
\end{equation}
where the diffusion approximation $\tilde{f}(x) = |x|^m Y_{m,j}(x/|x|)$ is the harmonic extension of $Y_{mj}$. Next, we estimate the solution $\phi_{mj}$ in~\eqref{eq:PERTURB} which satisfies the following RTE
\begin{equation}
v\cdot \nabla \phi_{mj} + \left(\eps \sigma_a + \frac{1}{\eps}\sigma_s\right) \phi_{mj} = \frac{1}{\eps}\sigma_s \aver{\phi_{mj}} - \eps\sigma_a u_{0,mj}
\end{equation}
with vacuum incoming boundary condition. For the source term on the right-hand-side, note that $\supp\sigma_a \subset K = B(0, r_0)$ and the estimate~\eqref{EQ: MJ} in $K\times\Omega$. Therefore $\phi_{mj} = \phi_{mj}^1 + \eps \phi_{mj}^2 + \eps^2 R_{mj}$ satisfies that~\cite{dautray1993mathematical}
\begin{equation}
\|\phi_{mj}^1 -  \Phi_{mj}(x) + \frac{\eps}{\sigma_s} v\cdot \nabla\Phi_{mj} \|_{L^2(D\times \Omega)} =\cO(  \eps^2) \|\tilde{f}\|_{L^2(D)},
\end{equation}
where the first part $\phi_{mj}^1$ is the solution to
\begin{equation}
v\cdot \nabla \phi_{mj}^1 + \left(\eps \sigma_a + \frac{1}{\eps}\sigma_s\right) \phi_{mj} = \frac{1}{\eps}\sigma_s \aver{\phi_{mj}} - \eps\sigma_a \tilde{f}
\end{equation}
and $\Phi_{mj}$ is the solution to the following diffusion equation,
\begin{equation}
\begin{aligned}
-\frac{1}{d}\nabla\cdot \left(\frac{1}{\sigma_s} \nabla \Phi_{mj} \right) + \sigma_a \Phi_{mj} &= -\sigma_a \tilde{f}\quad &\text{ in }& D, \\
\Phi_{mj} + \eps \ell  \partial_n\Phi_{mj}&= 0\quad&\text{ on }&\partial D.
\end{aligned}
\end{equation}
The parameter $\ell$ is the extrapolation length. The second part $\phi_{mj}^2$ satisfies the equation
\begin{equation}
    v\cdot \nabla \phi_{mj}^2 +  \left(\eps \sigma_a + \frac{1}{\eps}\sigma_s\right) \phi_{mj}^2   = \frac{1}{\eps}\sigma_s \aver{\phi_{mj}^2} - \eps \sigma_a v\cdot \nabla \tilde{f}
\end{equation}
so we can estimate $ \|\phi_{mj}^2(x, v)\| = \cO(\eps)\|\nabla \tilde{f}\|_{L^2(D)}$ from the standard expansion method which makes $\eps\phi_{mj}^2$ absorbed into the corrector term $\eps^2 R_{mj}$.
The corrector term $\eps^2 R_{mj}$ satisfies the following estimate
\begin{equation}\label{eq:COR}
\|R_{mj}\|_{L^2(D\times\Omega)}  \le C \|\sigma_a\|_{\infty} \left(\|\nabla^2 \tilde{f}\|_{L^{2}(D)} + \|\frac{\partial\tilde{f}}{\partial n}\|_{L^2(\partial D)}\right)
\end{equation}
for some constant $C$ independent of $\eps$. Let $J_{mj}$ denote the velocity averaged vector field, 
\begin{equation}
J_{mj}(x) = \int_{\Omega} v\phi_{mj}(x, v) d\mu(v) = -\frac{\eps}{d\sigma_s} \nabla \Phi_{mj}(x) + \eps^2 \tilde{R}_{mj},
\end{equation}
where the vector field $\tilde{R}_{mj}$ is 
\begin{equation}\label{eq:RES}
\tilde{R}_{mj} = \int_{\Omega} v R_{mj}(x, v) dv.
\end{equation}
Therefore $I_{2,mjnk}$ is bounded by
\begin{equation}\label{eq:EST2}
\begin{aligned}
|I_{2,mjnk}| &= \left| \frac{1}{\eps} \int_{D} \nabla (|x|^n \overline{Y_{nk}(x/|x|)} )\cdot J_{mj}(x) dx\right| \\&= \frac{1}{\eps}\left|\int_D \nabla (|x|^n \overline{Y_{nk}(x/|x|)})  \cdot \left(-\frac{\eps}{d\sigma_s} \nabla \Phi_{mj}(x) + \eps^2 \tilde{R}_{mj}\right)  dx\right|\\
&\le \eps \left| \int_D \nabla (|x|^n \overline{Y_{nk}(x/|x|)} ) \cdot  \tilde{R}_{mj} dx \right| + \eps \ell \left|\int_{\partial D} \partial_n (|x|^n \overline{Y_{nk}(x/|x|)})  \partial_n\Phi_{mj} ds \right|\\
& \le C \eps \|\sigma_{a}\|_{\infty} \left(\|\nabla^2 \tilde{f}\|_{L^{2}(D)} + \|\frac{\partial\tilde{f}}{\partial n}\|_{L^2(\partial D)}\right) \|\nabla (|x|^n \overline{Y_{nk}(x/|x|)})\|_{L^2(D)}\\&\quad +  \eps\ell \| \partial_n (|x|^n \overline{Y_{nk}(x/|x|)}) \|_{H^{1/2}(\partial D)} \| \partial_n \Phi_{mj}\|_{H^{-1/2}(\partial D)}.
\end{aligned}
\end{equation}
Here we have used integration by parts and the fact that $\Phi_{mj} = -\eps \ell \partial_n\Phi_{mj} $ on $\partial D$ in the second line of~\eqref{eq:EST2}. It is easy to find out there exists constant $C'$ that
\begin{equation}
\|\nabla (|x|^n \overline{Y_{nk}(x/|x|)})\|_{L^2(D)} < C' \sqrt{1 + n}.
\end{equation}
and $\left(\|\nabla^2 \tilde{f}\|_{L^{2}(D)} + \|\frac{\partial\tilde{f}}{\partial n}\|_{L^2(\partial D)}\right)\sim \|f\|_{H^{3/2}(\partial D)}$ is bounded by $O((1+m)^{3/2})$. For the second term, $ \| \partial_n \Phi_{mj}\|_{H^{-1/2}(\partial D)}\le C \|\sigma_a \tilde{f}\|_{L^2(D)}$ and $\| \partial_n (|x|^n \overline{Y_{nk}(x/|x|)}) \|_{H^{1/2}(\partial D)}\le C (1+n)^{3/2}$.  Hence we obtain
\begin{equation}
|I_{2, mjnk}|\le C \eps \|\sigma_a\|_{\infty} ({1 + n})^2.
\end{equation}
When $m \ge n$, $M_{2,mjnk} = L_{2, mjnk}$, we can use the adjoint RTE to acquire a similar estimate
\begin{equation}
|L_{2, mjnk}|\le C \eps \|\sigma_a\|_{\infty} (1 + m)^2.
\end{equation}
Combine the above two estimates, we conclude that $|M_{2,mjnk}|\le C \eps \|\sigma_a\|_{\infty} (1 + \max(m, n))^2$.
\end{proof}

\section{The instability estimate}\label{sec:INSTAB}
In this section, we prove the main theorem with respect to the instability estimate. Before that, we introduce the following definitions based on Kolmogorov's masterwork~\cite{kolmogorov1959}.
\begin{definition}
	Let $(X, \mathsf{d})$ be a metric space and  $\delta > 0$, then we say 
	\begin{enumerate}
		\item A set $Y \subset X$ is a $\delta$-net for $X_1\subset X$ if for any $x\in X_1$ there exists $y\in Y$ such that $\mathsf{d}(x, y)\le \delta$.
		\item  A set $Z\subset X$ is $\theta$-distinguishable if for any distinct $z_1,z_2\in Z$, we have $\mathsf{d}(z_1, z_2)\ge \theta$.
		
	\end{enumerate}
\end{definition}

The following lemma shows that the number of $q$-times differentiable functions grows at least exponentially with its $C^q$ norm. The proof can be easily adapted from~\cite{kolmogorov1959ϵ}.
\begin{lemma}[Kolmogorov]\label{lem:KOLMO}
	Let $d\ge 2$ and $q > 0$. For $\theta, \beta > 0$, consider the metric space
	\begin{equation}
	X_{q\theta\beta} = \{ f \in C_0^q (K): \|f\|_{\infty} \le \theta, \|f\|_{C^q}\le \beta \text{ and }f\ge 0\}
	\end{equation}
	the metric is induced by $L^{\infty}$ norm. Then there is a constant $\mu > 0$ such that for any $\beta > 0$ and $\theta \in (0, \mu\beta)$, there is a $\theta$-distinguishable set $Z\subset X_{q\theta\beta}$, its cardinality satisfies following lower bound.
	\begin{equation}
	|Z| \ge \exp\left(2^{-d-1}(\mu\beta/\theta)^{d/q} \right).
	\end{equation}
\end{lemma}
For any bounded linear operator $\cA:\cX\to\cY$, we consider its matrix representation by the entries $a_{mjnk} = \langle \cA Y_{mj} , Y_{nk}\rangle$, then the operator norm of $\cA$ can be bounded by
\begin{equation}\label{eq:NORM}
\begin{aligned}
\|\cA\|^2_{\cX\to \cY} \le &\, \sum_{m,j,n,k}   (1+m)^{-2s}(1+n)^{-2s} |a_{mjnk}|^2.
\end{aligned}
\end{equation}
The $4$-tuple of integers $(m,j,n,k)$ in the above summation runs through all combinations that $m , n \ge 0$, $1\le j \le p_m, 1\le k \le p_n$. 
\begin{lemma}
	$\|\cA\|^2_{\cX\to \cY} \le 32 \sup_{m,j,n,k} (1 + \max(m, n))^{d-2s} |a_{mjnk}|^2$.
\end{lemma}
\begin{proof}
	We separate the summation into two parts: $m > n$ and $m \le n$. Since the dimension of subscripts $j$ and $k$ are bounded by $2(1+m)^{d-2}$ and $2(1+n)^{d-2}$, we have 
	\begin{equation}\nonumber
	\begin{aligned}
	 \sum_{m,j,n,k}   (1+m)^{-2s}(1+n)^{-2s} |a_{mjnk}|^2 \le  4\sum_{m = 0}^{\infty} (1+m)^{d-2-2s} \sum_{n = 0}^{m-1} (1+n)^{d-2-2s} \sup_{n < m}|a_{mjnk}|^2 \\
	 + 4\sum_{n = 0}^{\infty} (1+n)^{d-2-2s} \sum_{m = 0}^{n} (1+m)^{d-2-2s} \sup_{m \le n}|a_{mjnk}|^2.
	\end{aligned}
	\end{equation}
	The first term stands for the summation of $m > n$, the second term stands for the summation of $m \le n$. The supremes are taken over all $1\le j\le p_m, 1\le k\le p_n$. Because $s \ge \frac{d}{2}$, $d - 2 - 2s\le -2$, the summation $ \sum_{n = 0}^{m-1} (1+n)^{d-2-2s} \le  \sum_{n = 0}^{m-1} (1+n)^{-2} < 2$,
	\begin{equation}
	4\sum_{m = 0}^{\infty} (1+m)^{d-2-2s} \sum_{n = 0}^{m-1} (1+n)^{d-2-2s} \sup_{m > n}|a_{mjnk}|^2 \le 	8\sum_{m = 0}^{\infty} (1+m)^{d-2s-2} \sup_{n < m}|a_{mjnk}|^2.
	\end{equation}
	Similarly, the other term is bounded by
	\begin{equation}
	4\sum_{n = 0}^{\infty} (1+n)^{d-2-2s} \sum_{m = 0}^{n} (1+m)^{d-2-2s} \sup_{m \le n}|a_{mjnk}|^2\le 	8\sum_{n = 0}^{\infty} (1+n)^{d-2s-2} \sup_{m \le n}|a_{mjnk}|^2.
	\end{equation}
	Therefore, we can combine the above two bounds 
	\begin{equation}\label{eq:NORM2}
	\begin{aligned}
		 \|\cA\|^2_{\cX\to \cY}&\le 16 \sum_{l = 0}^{\infty} (1 + l)^{d-2s-2} \sup_{l = \max(m, n)}|a_{mjnk}|^2 \\
		 &\le 16 \left(\sum_{l=0}^{\infty} (1 + l)^{-2} \right) \sup_{l}\left( (1 + l)^{d-2s} \sup_{l=\max(m, n)}|a_{mjnk}|^2\right)\\
		 & \le 32 \sup_{l}\left( (1+l)^{d-2s} \sup_{l = \max(m ,n)} |a_{mjnk}|^2 \right) \\
		 & = 32 \sup_{m,j,n,k} (1 + \max(m, n))^{d-2s} |a_{mjnk}|^2.
	\end{aligned}
	\end{equation}
\end{proof}
Define the following Banach space $X_s$ with $s > \frac{d+4}{2}$ by
\begin{equation}\nonumber
X_s := \{ (a_{mjnk})\, \big| \left\|(a_{mjnk})\right\|_{X_s} := \sup_{m,j,n,k} (1 + \max(m,n))^{d/2 - s} |a_{mjnk}| <\infty \},
\end{equation}
then the estimate~\eqref{eq:NORM2} shows that $\|\cA\|_{X\to Y}\le 4\sqrt{2} \left\|(a_{mjnk})\right\|_{X_s}$.
We define $B_{+, R}^{\infty}$ as,
\begin{equation}
B_{+, R}^{\infty} = \{ f\in L^{\infty}(K) | \|f\|_{\infty} \le R, f\ge 0 \},
\end{equation} 
then we have the following lemma about the $\delta$-net of $\Gamma(B_{+,R}^{\infty})$.
%
	\begin{lemma}\label{lem:NET}
		The operator $\Gamma$ maps $B_{+,R}^{\infty}$ into $X_s$ for any $s > \frac{d+4}{2}$. There exists a constant $0 < \eta = \eta(s, d)$ that for every $\delta \in (0, e^{-1})$, there is a $\delta$-net $Y$ for $\Gamma(B_{+,R}^{\infty})$ in $X_s$, with at most $\exp\left( \eta \left(\log \delta^{-1} + \frac{\eps}{\delta}+ \left(\frac{\eps}{\delta}\right)^{-1/\tau}\right)^{2d+1}\right)$ elements, where $\tau = \frac{d+4}{2} -s < 0$.
	\end{lemma}
	\begin{proof}
		For any $s > \frac{d+4}{2}$ and $\sigma_a\in B^{\infty}_{+, R}$, using Lemma~\ref{lem:ESTIMATE}, we have
		\begin{equation}
		\begin{aligned}
		\sup_{m,j,n,k}(1+\max(m,n))^{d/2-s}|(M_{1,mjnk})| &\le  C_0 R \sup_{l} (1+l)^{d/2-s + 1}  r_0^{l}  < \infty,\\
		\sup_{m,j,n,k}(1+\max(m,n))^{d/2-s}|(M_{2,mjnk})|&\le C_0 R \eps \sup_{l} (1 + l)^{d/2 - s + 2} \le C_0 R\eps < \infty.
		\end{aligned}
		\end{equation}
	then we have 
		$\big\|\Gamma(\sigma_a)\big\|_{X_s} <\infty,$
		which implies that $\Gamma(B_{+, R}^{\infty}) \subset X_s$. Denote $\tau = \frac{d+4}{2} -s < 0$ and let $l_{\delta s}$ be the smallest integer such that $\forall l\ge l_{\delta s}$ 
		\begin{equation}
		(1 + l)^{\tau} \left(r_0^{l} + \eps \right) \le \frac{\delta}{4C_0 R}.
		\end{equation}
		We conclude that
		\begin{equation}
		l_{\delta s} \le \max(l_1, l_2) \le l_1 + l_2,
		\end{equation}
		where $l_1$ and $l_2$ are the solutions to following equations.
		\begin{equation}
		(1 + l_1)^{\tau} r_0^{l_1} = \frac{\delta}{8C_0R},\quad (1+l_2)^{\tau}\eps = \frac{\delta}{8C_0R}.
		\end{equation}
		It is easy to deduce that $l_1 \le \log(\frac{\delta}{8C_0 R})/\log r_0$ and $l_2 = \left(\frac{8C_0 R\eps}{\delta}\right)^{-1/\tau}-1$, then there exists an absolute constant $\hat{C}$ that
		\begin{equation}\label{eq:TRUNC_ESTIMATE}
		1 + l_{\delta s} \le \hat{C}\left(\log \delta^{-1} + \left(\frac{\eps}{\delta}\right)^{-1/\tau}\right).
		\end{equation}
		For all $4$-tuples $(m,j,n,k)$ such that $\max(m,n)\le l_{\delta s}$, the upper bounds are given by Lemma~\ref{lem:ESTIMATE} that $|M_{1,mjnk}|\le C_0 C_0' R$ and $|M_{2,mjnk}| \le C_0 R\eps (1 + l_{\delta s})^2$, where $C_0'=\sup_{l} (1+l)r_0^l$. We denote $\delta' = \frac{\delta}{8\sqrt{2}}$ and consider the sets
		\begin{equation}
		\begin{aligned}
		Y_{1, \delta s}& := \delta'\bbZ \cap [-C_0C_0' R, C_0 C_0' R], \\ Y_{2,\delta s}& := \delta'\bbZ \cap [-C_0 R\eps (1 + l_{\delta s})^2, C_0 R\eps(1 + l_{\delta s})^2].
		\end{aligned}
		\end{equation}
		Then $|Y_{1,\delta s} |= 1 + 2\lfloor\frac{C_0 C_0'R}{\delta'} \rfloor\le 1 + \frac{16\sqrt{2} C_0C_0'R}{\delta}$ and $|Y_{2,\delta s}| = 1 + 2\lfloor \frac{C_0R\eps (1 + l_{\delta s})^2}{\delta'}\rfloor \le 1 + \frac{16\sqrt{2} C_0R\eps(1 + l_{\delta s})^2}{\delta}$. Define the following two sets
		\begin{equation}
		\begin{aligned}
		Y_1 &= \{ (b_{mjnk}) \,|\, b_{mjnk}\in Y_{1,\delta s} \text{ for } \max(m,n) \le l_{\delta s}, b_{mjnk} = 0\text{ otherwise} \}, \\
		Y_2 &= \{ (c_{mjnk}) \,|\, c_{mjnk}\in Y_{2,\delta s} \text{ for } \max(m,n) \le l_{\delta s}, c_{mjnk} = 0 \text{ otherwise}  \}.
		\end{aligned}
		\end{equation}
		We briefly prove $Y = Y_1 + Y_2$ is a $\delta$-net for $\Gamma(B_{+, R}^{\infty})$. Let $(a_{mjnk})\in \Gamma(B_{+, R}^{\infty})$,  we construct two elements $(b_{mjnk})\in Y_1$ and $(c_{mjnk})\in Y_2$ as approximations to $M_{1,mjnk}$ and $M_{2,mjnk}$. If $\max(m, n) \le l_{\delta s}$, we take $b'_{mjnk}$ to be one of the elements in $Y_{1,\delta s}$ closest to $M_{1,mjnk}$ and $c'_{mjnk}$ to be one of the  elements in $Y_{2,\delta s}$ closest to $M_{2,mjnk}$, then $|b'_{mjnk} - M_{1,mjnk}|\le \delta'$ and $|c'_{mjnk} - M_{2, mjnk}|\le \delta'$. If $\max(m, n) > l_{\delta s}$, then $b'_{mjnk} = c'_{mjnk} = 0$. Since $s > d/2$, $(1 + \max(m,n))^{d/2-s}\le 1$, then it is easy to obtain
		\begin{equation}\nonumber
		\begin{aligned}
		4 \sqrt{2}(1 + \max(m, n))^{d/2-s}(|b'_{mjnk} - M_{1,mjnk}| + |c'_{mjnk} - M_{2,mjnk}|) \le 8 \sqrt{2}\delta' = \delta.
		\end{aligned}
		\end{equation}
		For $\max(m, n) > l_{\delta s}$, above inequality also holds by the construction of $l_{\delta s}$. Therefore $Y$ is a $\delta$-net for $\Gamma(B^{\infty}_{+, R})$. 
		
		Next, we count the elements in $Y$. Since $|Y| = |Y_1| |Y_2|$, it remains to estimate the cardinalities for both $Y_1$ and $Y_2$. Let $n_l$ be the number of $4$-tuples $(m,j,n,k)$ such that $\max(m, n) = l$, then $Y_1$ has $|Y_{1,\delta s}|^{n_{\delta s}}$ elements and $Y_2$ has $|Y_{2, \delta s}|^{n_{\delta s}}$ elements, where $n_{\delta s} = \sum_{j = 0}^{l_{\delta s}} n_j \le 8(1+l_{\delta s})^{2d - 2}$. Using the fact that $\log(1+t)\le t$ for $t\ge 0$ and $\log\delta^{-1}\ge 1$, we can estimate $|Y|$ by applying~\eqref{eq:TRUNC_ESTIMATE} and taking $\eta$ sufficiently large,
		\begin{equation}
		\begin{aligned}
		|Y| &\le \left(\left(1 + \frac{16\sqrt{2}C_0 C_0'R}{\delta}\right) \left(1 + \frac{16\sqrt{2}C_0 R\eps (1 + l_{\delta s})^2}{\delta}\right)\right)^{8(1 + l_{\delta s})^{2d-2}}\\
		&\le \exp\left(\left( C_1 \log(\delta^{-1})+ \log\left(1 + \frac{16\sqrt{2}C_0 R\eps (1 + l_{\delta s})^2}{\delta}\right) \right)8(1 + l_{\delta s})^{2d-2}\right)\\
		& \le \exp\left( C_2 \left(\log(\delta^{-1}) +  \frac{\eps}{\delta}\right) (1 + l_{\delta s})^{2d}\right) \\
		&\le \exp\left( C_3 \left(\log(\delta^{-1}) +  \frac{\eps}{\delta}\right)\left(\log \delta^{-1} + \left(\frac{\eps}{\delta}\right)^{-1/\tau}\right)^{2d}\right)\\
		& \le \exp\left( \eta \left(\log \delta^{-1} + \frac{\eps}{\delta}+ \left(\frac{\eps}{\delta}\right)^{-1/\tau}\right)^{2d+1}\right).
		\end{aligned}
		\end{equation}
	\end{proof}
	\begin{proof}[Proof of theorem~\ref{thm:MAIN}]
		Let $\theta\in (0, \frac{R}{2})$, by Lemma~\ref{lem:KOLMO}, $\sigma_{a, 0} + X_{q\theta\beta}$ has a $\theta$-distinguishable set $\sigma_{a,0} + Z\subset \sigma_{a,0} + X_{q\theta\beta}$, then any two elements in $\sigma_{a, 0} + Z$ are separated by a distance of at least $\theta$ in $L^{\infty}$ norm and $\sigma_{a, 0} + X_{q\theta\beta} \subset B_{+,R}^{\infty}$. By Lemma~\ref{lem:NET}, the constructed  set $Y$ is a $\delta$-net for $\Gamma(\sigma_{a,0} + X_{q\theta\beta})$. When $|\sigma_{a, 0} + X_{q\theta\beta}| > |Y|$, then there are two absorption coefficients $\sigma_{a, 1}, \sigma_{a,2}\in \sigma_{a,0}+X_{q\theta\beta}$, their images under $\Gamma$ are in the same $X_s$-ball of radius $\delta$ centered at some element in $Y$, then use~\eqref{eq:NORM2}, we obtain
		\begin{equation}\label{eq:BOUND}
		\|\Gamma(\sigma_{a,1}) -\Gamma(\sigma_{a,2}) \|_{\cX\to\cY} \le 8\sqrt{2}\delta.
		\end{equation}
		We take $\delta$ be the unique solution to the following equation
		\begin{equation}
		\theta^{-\frac{d}{(2d+1)q}} = \log \delta^{-1} + \left(\frac{\eps}{\delta}\right) +\left( \frac{\eps}{\delta}\right)^{-1/\tau},\quad \tau = \frac{d+4}{2} -s < 0,
		\end{equation}
		and choose $\beta$ that
		\begin{equation}
		\beta > \mu^{-1}\max\left(\frac{R}{2}, \left(2^{(d+1)}\eta\right)^{q/d}\right),
		\end{equation}
		then $\mu\beta \ge \frac{R}{2} > \theta$ satisfies the requirement for Lemma~\ref{lem:KOLMO}, which implies 
		\begin{equation}
		|\sigma_{a,0} + X_{q\theta\beta}| = |X_{q\theta\beta}| \ge \exp(2^{-d-1}(\mu\beta/\theta)^{d/q}) > \exp(\eta \theta^{-d/q}) \ge |Y|.
		\end{equation}
	\end{proof}

\begin{proof}[Proof of Corollary~\ref{cor:EXP} and Corollary~\ref{cor:HOLDER}]
	Under the assumption that $s - \frac{d+4}{2} = 1$, it is clear that $\delta$  solves equation
	\begin{equation}
\theta^{-\frac{d}{(2d+1)q}}= \log\delta^{-1} + 2\left( \frac{\eps}{\delta}\right).
	\end{equation}
	which is exactly the function $\omega$ in Theorem~\ref{thm:MAIN}. We consider two cases. The first case, $\eps$ is sufficiently small,
	\begin{equation}
	\begin{aligned}
		&\left( \frac{\eps}{\delta}\right) \le \log\delta^{-1} \Rightarrow \eps \le \delta \log\delta^{-1}, \\
		&\theta^{-\frac{d}{(2d+1)q}} \le 3\log\delta^{-1} \Rightarrow 	\delta \le \exp\left( -\frac{1}{3} \theta^{-\frac{d}{(2d+1)q}}\right).
	\end{aligned}
	\end{equation}
The second case, $\eps$ is not sufficiently small, we have
	\begin{equation}
	\begin{aligned}
	&\left( \frac{\eps}{\delta}\right) \ge \log\delta^{-1} \Rightarrow \eps \ge \delta \log\delta^{-1}, \\
	&\theta^{-\frac{d}{(2d+1)q}}\le 3\left(\frac{\eps}{\delta}\right) \Rightarrow \delta \le 3\eps \theta^{\frac{d}{(2d+1)q}}.
	\end{aligned}
	\end{equation}
	The rest of proof is straightforward by applying Theorem~\ref{thm:MAIN}.
\end{proof}
\begin{remark}
	We also want to point out that the H\"{o}lder type instability estimate in~\cite{chen2018stability} is similar to our Corollary~\ref{cor:HOLDER} as special cases. For the special case $s - \frac{d+4}{2} = 1$ and a fixed perturbation $\delta$ in the measurement operator $\Lambda$, if the Knudsen number $\eps$ is large compared to the perturbation in the sense $\eps \gg  O(\delta \log \delta^{-1})$, the contribution of diffusion approximation $M_{1,mjnk}$ is dominated by that of transport $M_{2,mjnk}$. {Especially if $\eps = \cO(1)$, although the Hilbert expansion is no longer valid, the estimate of $M_{2,mjnk}$ is still true from standard transport theory, which will always dominate the $M_{1,mjnk}$.} The authors in~\cite{chen2018stability} linearize the transport equation which takes into account the transport term $M_{2,mjnk}$ while neglecting the diffusion term $M_{1,mjnk}$. Hence the instability estimate only contains H\"{o}lder type. However, when $\eps$ is sufficiently small such that $\eps \ll O(\delta \log\delta^{-1})$, then $M_{1,mjnk}$ becomes dominant and therefore can not be simply dropped.
\end{remark}
\section{Conclusion}\label{sec:CONCL}
In this paper we study the instability of {reconstruction of $\sigma_a$} in radiative transfer equation~\eqref{eq:RTE} with angularly averaged albedo operator measurement near diffusion limit $0 < \eps \ll 1$. When $\eps \to 0$, the problem degenerates to the inverse problem of a diffusion equation which is equivalent to the EIT problem. Our instability estimate characterizes the transition of instability from the balance of diffusion and transport. When $\eps$ is away from zero and perturbation in measurement $\delta$ is small enough, H\"{o}lder type of instability is observed, which is still ill-posed unless sufficient regularity of $\sigma_a$ is imposed.  Otherwise the exponential type of instability is observed, which is similar to the EIT problem. 
\section*{Acknowledgment}
H. Zhao is partially supported by NSF Grant DMS 1622490 and DMS 1821010. 

\appendix
\section{Appendix}
In this appendix, we briefly prove the following two key lemmas used in the proof of Theorem \ref{thm:MAIN}.
\begin{lemma}\label{lem:A1}
The integral operators $\cK_1$ and $\cK_2$ defined in~\eqref{eq:Ks}, satisfy
\begin{equation}
\|(I - \cK_1)^{-1}\cK_2 f\|_{L^2(D)} < C \|f\|_{H^{3/2}(\partial D)}
\end{equation} 
where $0 < \eps < 1$ and $C$ is a constant independent of $\eps$.
\end{lemma}
\begin{proof}
		We only have to prove for the case $\sigma_a \equiv 0$.  First, we estimate the norm of $(I - \cK_1)^{-1}$. The operator $\cK_1$ is defined as
		\begin{equation}
		\cK_1 f(x) = \frac{1}{\nu_d}\int_{D} \frac{\exp(-\frac{\sigma_s}{\eps}|x-y|)}{|x-y|^{d-1}}\frac{\sigma_s}{\eps}f(y)dy
		\end{equation}
		where $\sigma_s$ is a positive constant, $\nu_d$ is the area of the unit sphere in $\mathbb{R}^d$, and  $D$ is the unit ball in $\mathbb{R}^d$.  With a slight abuse of notation, we consider $f(x)$ as $f(x)\chi_{D}$, where $\chi_D$ is the characteristic function of $D$. Then
		\begin{equation}
		\cK_1 f = k\ast f.
		\end{equation}
		Take Fourier transform,
		\begin{equation}
		\widehat{\cK_1 f}(\xi) = \widehat{k}(\xi) \widehat{f}(\xi),
		\end{equation}
		where $k(x)$ is defined as
		\begin{equation}
		k(x) = \frac{1}{\nu_d}\frac{\exp(-\frac{\sigma_s}{\eps}|x|)}{|x|^{d-1}} \frac{\sigma_s}{\eps}.
		\end{equation}
		The Fourier transform of $k(x)$ is 
		\begin{equation}
		\begin{aligned}
		\widehat{k}(\xi)&=\frac{\nu_{d-1}}{\nu_d} \int_{0}^{\infty}\int_0^{\pi} \exp(-\frac{\sigma_s}{\eps}\rho)\frac{\sigma_s}{\eps} \exp(-i|\xi| \rho \cos\theta) d\rho \sin^{d-2} \theta d\theta \\ &= \frac{\nu_{d-1}}{\nu_d} \int_0^{\pi}\frac{\sin^{d-2}\theta d\theta}{1 + \eps^2 \frac{|\xi|^2}{\sigma_s^2}\cos^2\theta } .
		\end{aligned}
		\end{equation}
		It is easy to see that $|\widehat{k}(\xi)|<1$ and is a decreasing function of $|\xi|$. On the other hand, since $f$ is compactly supported in $D$, 
		\begin{equation}
		|\widehat{f}(\xi)| = |\int_{D} \exp({-ix\cdot\xi}) f(x) dx|\le \int_{D} |f(x)| dx \le C_D \|f\|_{L^2}
		\end{equation}
		for some $C_D$ depending on $D$ only. Then there exists an absolute constant $s > 0$ that
		\begin{equation}
		\int_{|\xi| < s} |\widehat{f}(\xi)|^2 d\xi \le \frac{1}{2}\|f\|_{L^2}^2.
		\end{equation}
		For $|\xi|\le s$ and $0 < \eps < 1$, we have
		\begin{equation}
		\begin{aligned}
	\widehat{k}(\xi) &=\frac{2 \nu_{d-1}}{\nu_d}\left( \int_0^{\hat{\theta}}\frac{\sin^{d-2}\theta d\theta}{1 + \eps^2 \frac{|\xi|^2}{\sigma_s^2}\cos^2\theta }+
	\int_{\hat{\theta}}^{\frac{\pi}{2}}\frac{\sin^{d-2}\theta d\theta}{1 + \eps^2 \frac{|\xi|^2}{\sigma_s^2}\cos^2\theta }\right) \\ &\le \frac{1}{2}\frac{1}{1 + \eps^2 \frac{|\xi|^2}{\sigma_s^2}\cos^2\hat{\theta }}+\frac{1}{2}
	\le \frac{1}{1 + C(\eps^2 \frac{|\xi|^2}{\sigma_s^2})},
		\end{aligned}
		\end{equation}
		where $C$ is positive constant independent of $\eps$ and $\hat{\theta}$ satisfies
		\[
		\frac{\nu_{d-1}}{\nu_d} \int_0^{\hat{\theta}}\sin^{d-2}\theta d\theta =\frac{\nu_{d-1}}{\nu_d} \int_{\hat{\theta}}^{\frac{\pi}{2}} 
		\sin^{d-2}\theta d\theta =\frac{1}{4}.
		\] 		
		Then we obtain the following estimate 
		\begin{equation}\nonumber
		\begin{aligned}
		\|\cK_1 f\|^2_{L^2} &= \|\widehat{\cK_1 f}\|^2_{L^2} = \|\widehat{k}\widehat{f}\|^2_{L^2} \le 
		\int_{|\xi| < s} |\widehat{k}|^2 |\widehat{f}|^2 + \int_{|\xi| \ge s} |\widehat{k}|^2 |\widehat{f}|^2  \\ &\le
		\int_{|\xi| < s} |\widehat{f}|^2  + \frac{1 }{1 + 2C\epsilon^2 \frac{s^2}{\sigma_s^2}}\int_{|\xi| \ge s}  |\widehat{f}|^2 \le \frac{1 + C\epsilon^2 \frac{s^2}{\sigma_s^2}}{1 + 2 C \epsilon^2\frac{s^2}{\sigma_s^2}} \|f\|_{L^2}^2 = (1 - O(\epsilon^2)) \|f\|^2_{L^2}.
		\end{aligned}
		\end{equation}
		Therefore $\|(I - \cK_1)^{-1}\|_{L^2(D)\to L^2(D)} \le O(\eps^{-2})$. In the next, we compute $\cK_2 f$. Denote $$
		G(r) = -\frac{1}{\nu_d}\int_{r}^{\infty} \frac{e^{-\frac{\sigma_s}{\eps}\rho}}{\rho^{d-1}} d\rho,$$
		then we can easily verify that
		\begin{equation}
		\begin{aligned}
		\cK_2 f &= \int_{\partial D} \partial_{\bn} G(|x - y|) f(y) dS(y) \\
		&= \int_{D} \nabla \cdot (\nabla G(|x - y|) \tilde{f}(y)) dy \\
		&= (I - \cK_1) \tilde{f} + \int_{D} \nabla G(|x-y|)\cdot \nabla \tilde{f}(y) dy \\
		&= (I -\cK_1) \tilde{f} - \int_{D} G(|x - y|)\Delta \tilde{f}(y)dy + \int_{\partial D} G(|x-y|) \frac{\partial \tilde{f}}{\partial n} dS(y),
		\end{aligned}
		\end{equation}
		where $\tilde{f}$ is an extension of $f$ in $D$ such that $ \tilde{f}|_{\partial D} = f,  \partial_{n}\tilde{f}|_{\partial D} = 0$ and $\|\tilde{f}\|_{L^2(D)} + \|\Delta \tilde{f}\|_{L^2(D)} \le \tilde{C} \|f\|_{H^{3/2}(\partial D)}$. Such an $\tilde{f}$ in the unit ball can be explicitly constructed using spherical harmonics. 
		Since $\|G\|_{L^2(D)\to L^2(D)} = O(\eps^2)$, we have 
		\begin{equation}\nonumber
		\|(I - \cK_1)^{-1}\cK_2 f\|_{L^2(D)} = \|\tilde{f} - (I-\cK_1)^{-1} G\ast\Delta\tilde{f}\|_{L^2(D)}
		\le \hat{C} (\|\tilde{f}\|_{L^2(D)} + \|\Delta \tilde{f}\|_{L^2(D)}) \le C \|f\|_{H^{3/2}(\partial D)}.
		\end{equation}
\end{proof}
\begin{lemma}\label{lem:A2}
	Suppose $u$ is the solution to~\eqref{eq:RTE}, then
	\begin{equation}
	\|u\|_{L^2(D\times\Omega)} \le C \|f\|_{H^{3/2}(\partial D)}
	\end{equation}
	where $0 < \eps < 1$ and $C$ is independent of $\eps$.
\end{lemma}
\begin{proof}
	We estimate the solution  by solving the linear transport equation with no absorption,
	\begin{equation}\nonumber
	|u(x, v)| \le \int_{0}^{\tau_{-}(x, v)} \exp\left(-\frac{\sigma_s}{\eps} l\right) \frac{\sigma_s}{\eps}\left|\aver{u}(x - lv)\right| dl + \exp\left(-\frac{\sigma_s}{\eps} \tau_{-}(x, v)\right)\left| f(x - \tau_{-}(x, v) v)\right|
	\end{equation} 
	where $\tau_{-}(x, v) = \sup\{t\, |\, x - tv \in D \}$.  Denote line integral operator $T$ by $$Tg(x,v) := \int_{0}^{\tau_{-}(x, v)} \exp\left(-\frac{\sigma_s}{\eps} l\right) \frac{\sigma_s}{\eps}g(x - lv) dl,$$ then we have the following $L^{\infty}$ and $L^1$ estimates for $Tg$,
	\begin{equation}\nonumber
	\begin{aligned}
\left|Tg(x, v) \right| &\le \|g\|_{\infty} \int_0^{\tau_{-}(x, v)} \exp\left(-\frac{\sigma_s}{\eps} l\right) \frac{\sigma_s}{\eps} dl < \|g\|_{\infty},\\
\int_{D\times\Omega} |Tg(x, v) |dx dv &= \int_{D\times D} \frac{\exp(-\frac{\sigma_s}{\eps}|x-y|)}{|x-y|^{d-1}}\frac{\sigma_s}{\eps} |g(y)| dx dy < \|g\|_{L^1}.
\end{aligned}
\end{equation}
The latter inequality is from Young's inequality.
By Riesz-Thorin interpolation and Lemma \ref{lem:A1}, when $0<\eps<  1$, there exists a constant $c_1$ independent of $\eps$ that
\begin{equation}\label{eq:T}
\|T\aver{u}\|_{L^2(D\times\Omega)} < \|\aver{u}\|_{L^2(D)} \le c_1 \|f\|_{H^{3/2}(\partial D)}.
\end{equation}
It remains to estimate the boundary contribution term $S h(x, v) = \exp\left(-\frac{\sigma_s}{\eps} \tau_{-}(x, v)\right) h(x - \tau_{-}(x, v) v)$, we follow a similar approach
\begin{equation}
|Sh(x, v)| \le \|h\|_{\infty}
\end{equation}
and by Young's inequality, there is a constant $c_2$ independent of $\eps$ when $0 < \eps <  1$ such that
\begin{equation}
\begin{aligned}
\int_{D\times\Omega} Sh(x, v) dx dv &= \int_D \int_{\partial D} \frac{\exp(-\frac{\sigma_s}{\eps}|x-y|)}{|x-y|^{d-1}} \left|n_y\cdot \frac{x-y}{|x-y|}\right| h(y) dS(y) dx\\
&\le c_2 \|h\|_{L^1(\partial D)}.
\end{aligned} 
\end{equation}
Then $\|S f\|_{L^2(D\times\Omega)} \le \sqrt{c_2} \|f\|_{L^2(\partial D)}$ by Riesz-Thorin interpolation. 
\end{proof}
\bibliographystyle{siam}
\bibliography{main}

\end{document}